\documentclass[12pt,american,aps,prl,reprint]{revtex4-2}
\usepackage[T1]{fontenc}
\usepackage[latin9]{inputenc}
\usepackage{color}
\usepackage{babel}
\usepackage{amsmath}
\usepackage{amsthm}
\usepackage{graphicx}
\usepackage[unicode=true,pdfusetitle,
 bookmarks=true,bookmarksnumbered=false,bookmarksopen=false,
 breaklinks=false,pdfborder={0 0 0},pdfborderstyle={},backref=false,colorlinks=true]
 {hyperref}

\makeatletter
\theoremstyle{remark}
\newtheorem*{notation*}{\protect\notationname}
\theoremstyle{plain}
\newtheorem{thm}{\protect\theoremname}
\theoremstyle{definition}
\newtheorem{defn}[thm]{\protect\definitionname}
\theoremstyle{plain}
\newtheorem{lem}[thm]{\protect\lemmaname}
\theoremstyle{definition}
\newtheorem{example}[thm]{\protect\examplename}

\makeatother

\providecommand{\definitionname}{Definition}
\providecommand{\examplename}{Example}
\providecommand{\lemmaname}{Lemma}
\providecommand{\notationname}{Notation}
\providecommand{\theoremname}{Theorem}

\begin{document}
\title{Dynamical Entanglement}
\author{Gilad Gour}
\email{gour@ucalgary.ca}

\affiliation{Department of Mathematics and Statistics, University of Calgary, AB,
Canada T2N 1N4}
\affiliation{Institute for Quantum Science and Technology, University of Calgary,
AB, Canada T2N 1N4}
\author{Carlo Maria Scandolo}
\email{carlomaria.scandolo@ucalgary.ca}

\affiliation{Department of Mathematics and Statistics, University of Calgary, AB,
Canada T2N 1N4}
\affiliation{Institute for Quantum Science and Technology, University of Calgary,
AB, Canada T2N 1N4}
\begin{abstract}
Unlike the entanglement of quantum states, very little is known about
the entanglement of bipartite channels, called dynamical entanglement.
Here we work with the partial transpose of a superchannel, and use
it to define computable measures of dynamical entanglement, such as
the negativity. We show that a version of it, the max-logarithmic
negativity, represents the exact asymptotic dynamical entanglement
cost. We discover a family of dynamical entanglement measures that
provide necessary and sufficient conditions for bipartite channel
simulation under local operations and classical communication and
under operations with positive partial transpose.
\end{abstract}
\maketitle

\paragraph{Introduction.}

Quantum entanglement \citep{Plenio-review,Review-entanglement} is
universally regarded as the most important quantum phenomenon, signaling
the definitive departure from classical physics \citep{Schrodinger}.
Its importance ranges across different areas of physics, from quantum
thermodynamics \citep{Bocchieri,Seth-Lloyd,Lubkin,Gemmer-Otte-Mahler,Canonical-typicality,Popescu-Short-Winter,Mahler-book,Huber1,Huber2,delRio,Zurek},
to quantum field theory \citep{Ryu1,Ryu2,Witten} and condensed matter
\citep{Entanglement-many-body,Area-law,Condensed-matter}. In quantum
information it is a resource in many protocols that cannot be implemented
in classical theory, such as quantum teleportation \citep{Teleportation},
dense coding \citep{Dense-coding}, and quantum key distribution \citep{Ekert}.

An even more crucial aspect of physics is that all systems evolve.
This is described by quantum channels \citep{Nielsen2010,Wilde,Watrous}.
Given the importance of entanglement, a natural question is how physical
evolution interacts with it. For example, one can wonder how much
entanglement a given evolution creates or consumes.

To this end, in this letter, which is a concise presentation of the
most significant results of our previous work \citep{Gour-Scandolo},
we push entanglement out of its boundaries to the next level: from
quantum states (static entanglement) to quantum channels (dynamical
entanglement), filling an important gap in the literature (an independent
work in this respect is Ref.~\citep{Wilde-entanglement}). Preliminary
work was done in Refs.~\citep{Bennett-bipartite,Berta-cost,Bennett-converse,Pirandola-LOCC,Wilde-cost,Wilde-PPT-Das1,WW18},
but here we study the topic in utmost generality, using \emph{resource
theories} \citep{Quantum-resource-1,Quantum-resource-2,Resource-knowledge,Resource-currencies,Gour-single-shot,Regula2017,Multiresource,Gour-review,Adesso-resource,Single-shot-new}.
With them, the idea of entanglement as a resource can be made precise.
Resource theories have recently attracted considerable attention \citep{Gour-review},
producing plenty of important results in quantum information \citep{Plenio-review,Review-entanglement,delRio,Lostaglio-thermo,Review-coherence,Veitch_2014,Magic-states}.
Resource theories are particularly meaningful whenever there is a
restriction on the set of quantum operations that can be performed,
usually coming from the physical constraints of a task an agent is
trying to do \citep{Gour-review}.

Looking closely at the entanglement protocols mentioned above \citep{Teleportation,Dense-coding},
one notices that a state is converted into a particular channel \citep{Devetak-Winter,Resource-calculus}.
Thus, the need of a framework that goes beyond the conversion between
static entangled resources is built in the very notion of entanglement
as a resource. In other terms, we want to treat static and dynamical
resources on the same grounds. We do so by phrasing entanglement theory
as a resource theory of\emph{ }quantum processes \citep{Gour-review,Resource-channels-1,Resource-channels-2,Gour-Winter}.
In this setting, the generic resource is a bipartite channel \citep{Shannon-bipartite,Bipartite},
instead of a bipartite state.

In this letter, we start from the simulation of bipartite channels
with local operations and classical communication (LOCC) \citep{LOCC1,LOCC2,Lo-Popescu},
and we derive a family of convex dynamical entanglement measures that
provide necessary and sufficient conditions for the LOCC-simulation
of channels.

The key tool for the remainder of the letter is a generalization of
partial transpose \citep{PPT-Peres,PPT-Horodecki}. This allows us
to define superchannels with positive partial transpose (PPT) \citep{Leung},
which constitute the largest set of superoperations to manipulate
dynamical entanglement, also encompassing the standard entanglement
manipulations involving LOCC. In this setting, we define measures
of dynamical entanglement that can be computed efficiently with semidefinite
programs (SDPs). Specifically, one of them, the max-logarithmic negativity,
quantifies the amount of static entanglement needed to simulate a
channel using PPT superchannels.

Finally, with the same generalization of the partial transpose, we
discover bound dynamical entanglement, whereby it is not possible
to produce entanglement out of a class of channels---PPT channels
\citep{PPT1,PPT2}---that generalize PPT states \citep{PPT-Peres,PPT-Horodecki}.
\begin{notation*}
Physical systems are denoted by capital letters (e.g., $A$) with
$AB$ meaning $A\otimes B$. Working on quantum channels, it is convenient
to associate two subsystems $A_{0}$ and $A_{1}$ with every system
$A$, referring, respectively, to the input and output of the resource.
In the case of static resources, we take $A_{0}$ to be one-dimensional.
A channel from $A_{0}$ to $A_{1}$ is indicated with a calligraphic
letter $\mathcal{N}_{A}:=\mathcal{N}_{A_{0}\rightarrow A_{1}}$. Superchannels
are denoted by capital Greek letters (e.g., $\Theta$), and the action
of superchannels on channels by square brackets. Thus $\Theta_{A\rightarrow B}\left[\mathcal{N}_{A}\right]$
indicates the action of the superchannel $\Theta$ on the channel
$\mathcal{N}_{A}$.
\end{notation*}

\paragraph{LOCC simulation of bipartite channels.}

To manipulate dynamical resources, one needs quantum superchannels
\citep{Chiribella2008,Gour2018}, which are linear maps sending quantum
channels to quantum channels in a complete sense, i.e., even when
tensored with the identity superchannel. This means that if $\mathcal{N}_{RA}$
is a quantum channel, $\Theta_{A\rightarrow B}\left[\mathcal{N}_{RA}\right]$
is still a quantum channel, for any $R$.  Superchannels can be all
realized concretely with a preprocessing channel and a postprocessing
channel, connected by a memory system \citep{Chiribella2008,Gour2018}.
Specifically, an LOCC superchannel, used in LOCC simulation, consists
of LOCC pre- and post-processing, and is represented in Fig.~\ref{fig:LOCC-1}.
\begin{figure}
\begin{centering}
\includegraphics[scale=0.35]{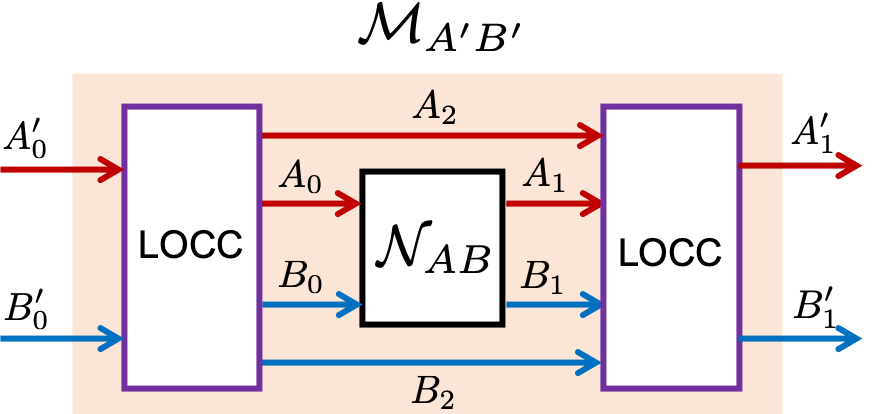}
\par\end{centering}
\caption{\label{fig:LOCC-1}Simulation of the channel $\mathcal{M}_{A'B'}$
from the channel $\mathcal{N}_{AB}$ with an LOCC superchannel, which
has LOCC pre- and postprocessing. Notice the presence of a memory
system for each of the two parties ($A_{2}$ and $B_{2}$, respectively).}
\end{figure}
 These superchannels are relevant when one is concerned with channel
simulation in bipartite communication-type scenarios where only classical
communication is allowed between the parties \citep{LOCC1,LOCC2}
(e.g., in teleportation \citep{Teleportation}).

Recall that with one qubit maximally entangled state (also known as
an ebit), thanks to quantum teleportation \citep{Teleportation},
we can simulate a qubit noiseless channel from Alice to Bob using
an LOCC scheme, and vice versa \citep{Devetak-Winter,Resource-calculus}.
Therefore one ebit---a static resource---is equivalent to a dynamical
one: a qubit channel. With a pair of such channels at hand, from Alice
to Bob and vice versa, we can LOCC-implement all bipartite channels
between the two parties when they both have qubit systems. This means
that such a pair of channels is the maximal resource. This is illustrated
in Fig.~\ref{fig:swap}(a).
\begin{figure}
\begin{centering}
\includegraphics[width=1\columnwidth]{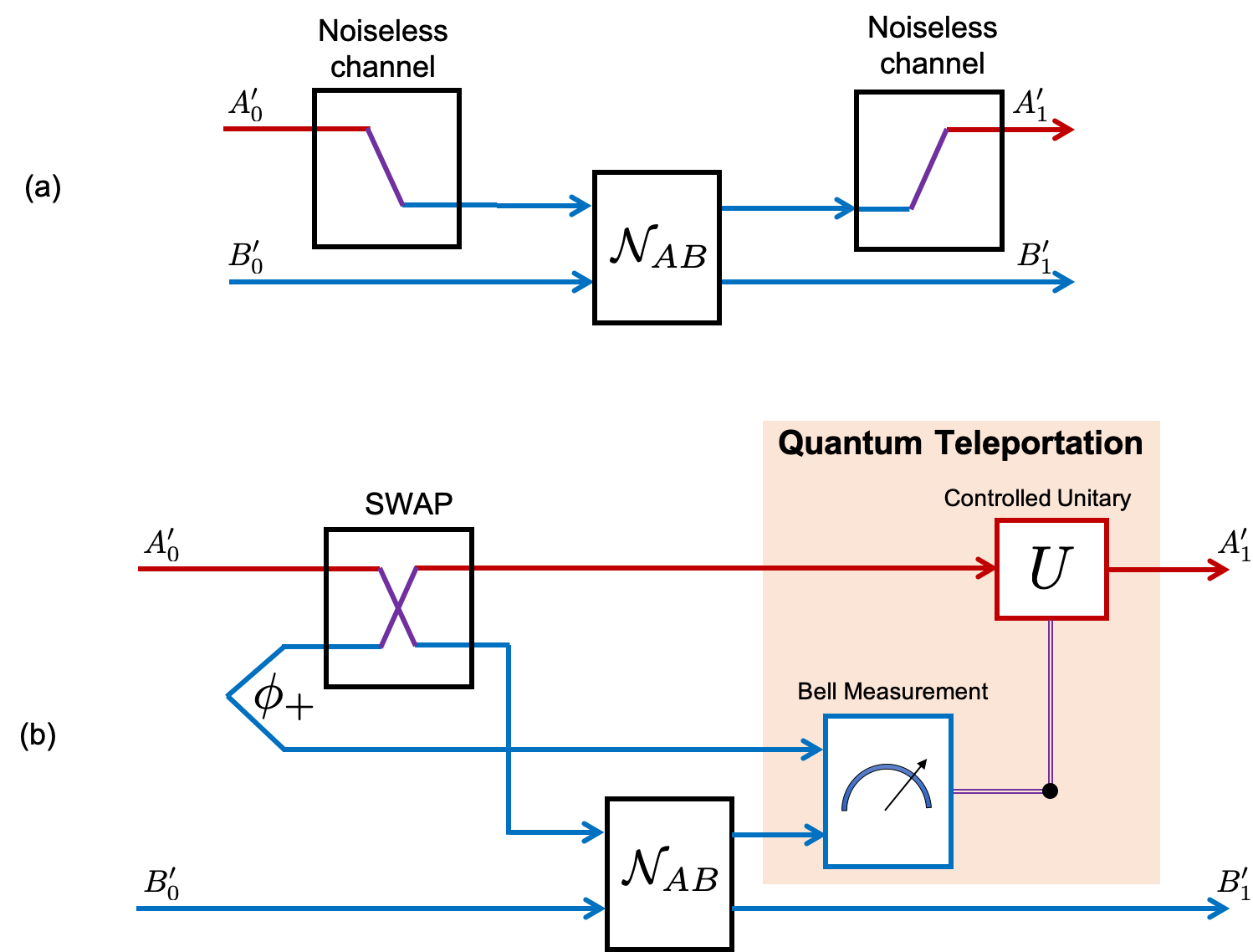}
\par\end{centering}
\caption{\label{fig:swap}(a) Simulation of an arbitrary bipartite channel
$\mathcal{N}_{A'B'}$ with two noiseless channels from Alice to Bob
and vice versa. (b) Simulation of an arbitrary bipartite channel $\mathcal{N}_{A'B'}$
with the swap resource and an LOCC superchannel (the postprocessing
is shaded in pink).}
\end{figure}
 In Fig.~\ref{fig:swap}(b) we show that in the same situation the
swap operation is another maximal resource, equivalent to 2 ebits.

In entanglement theory, a function $f$ is a \emph{measure of dynamical
entanglement} if $f\left(\Theta\left[\mathcal{N}_{AB}\right]\right)\leq f\left(\mathcal{N}_{AB}\right)$,
where $\Theta$ is an LOCC superchannel. It is conventional to assume
that $f$ vanishes on all separable channels \citep{SEP,Rains-SEP,PPT1},
which are regarded as the free resources in the theory of dynamical
entanglement; but this is not essential.

The very definition of a measure of dynamical entanglement indicates
that $f$ gives us a necessary condition for the simulation of channel
$\mathcal{M}$ starting from channel $\mathcal{N}$ and using an LOCC
superchannel $\Theta$. Indeed, if such a superchannel exists, namely
$\mathcal{M}=\Theta\left[\mathcal{N}\right]$, then $f\left(\mathcal{M}\right)\leq f\left(\mathcal{N}\right)$.
However, here we construct a family of convex measures of dynamical
entanglement that also give us a sufficient condition for LOCC simulation.
For any bipartite channels $\mathcal{P}$ and $\mathcal{N}$, define
\begin{equation}
E_{\mathcal{P}}\left(\mathcal{N}\right)=\sup_{\Theta\textrm{ LOCC}}\mathrm{Tr}\left[J^{\mathcal{P}}J^{\Theta\left[\mathcal{N}\right]}\right],\label{eq:complete family}
\end{equation}
where $J$ denotes the Choi matrix of the channel in the superscript,
and $\Theta$ is a generic LOCC superchannel. Note that these functions
need not vanish on separable channels. It is possible to show that
each function $E_{\mathcal{P}}$, with $\mathcal{P}$ ranging over
all bipartite channels, can be computed using a conic linear program
\citep[subsection 3 C]{Gour-Scandolo}.
\begin{thm}
In the theory of dynamical entanglement, a channel $\mathcal{N}$
can be LOCC-converted into a channel $\mathcal{M}$ if and only if
$E_{\mathcal{P}}\left(\mathcal{N}\right)\geq E_{\mathcal{P}}\left(\mathcal{M}\right)$
for \emph{every} bipartite channel $\mathcal{P}$.
\end{thm}

The proof is in Ref.~\citep[subsection 3 C]{Gour-Scandolo}. Since
we need to consider \emph{all} bipartite channels $\mathcal{P}$,
this family of measures of dynamical entanglement is not so practical
to work with. Unfortunately, one cannot expect to find a finite family
of such monotones, as shown in Ref.~\citep{Gour-infinite}.

Given two channels $\mathcal{N}$ and $\mathcal{M}$, to determine
if the former can be LOCC-converted into the latter, we can alternatively
compute their conversion distance, defined following similar ideas
to Ref.~\citep{Tomamichel}:
\begin{equation}
d_{\mathrm{LOCC}}\left(\mathcal{N}\to\mathcal{M}\right)=\frac{1}{2}\inf_{\Theta\textrm{ LOCC}}\left\Vert \Theta\left[\mathcal{N}\right]-\mathcal{M}\right\Vert _{\diamond}.\label{eq:conversion distance}
\end{equation}
If this distance is zero, we can convert $\mathcal{N}$ into $\mathcal{M}$
using a superchannel in the topological closure of LOCC superchannels
\citep{Chitambar2014}. Again, this distance can be calculated using
a conic linear program \citep[subsection 3 D]{Gour-Scandolo}, thanks
to the results in Refs.~\citep{Watrous-diamond,Gour-Winter}.

\paragraph{PPT superchannels.}

In entanglement theory, one of the most practical tools to determine
whether a state is entangled is the partial transpose \citep{PPT-Peres,PPT-Horodecki}.
One defines PPT states as the bipartite states $\rho_{AB}$ such that
$\mathcal{T}_{B}\left(\rho_{AB}\right)=\rho_{AB}^{T_{B}}$ is still
a valid state, where $\mathcal{T}$ denotes the transpose map \citep{PPT-Peres,PPT-Horodecki}.
Recall, however, that the set of PPT states is larger than the set
of separable states, due to the existence of bound entangled states
\citep{Bound-entanglement}. One then defines a PPT channel to be
a bipartite channel $\mathcal{N}_{AB}$ such that, applying the transpose
map on Bob's input and output, we get another valid channel $\mathcal{N}_{AB}^{\Gamma}:=\mathcal{T}_{B_{1}}\circ\mathcal{N}_{AB}\circ\mathcal{T}_{B_{0}}$
\citep{PPT1,PPT2}. Note that the set of PPT channels is larger than
the set of LOCC channels. It is not hard to show that PPT channels
are also completely PPT preserving, for they preserve PPT states even
when tensored with the identity channel \citep{PPT1,PPT2}. Finally,
the Choi matrix of a PPT channel $\mathcal{N}_{AB}$ is such that
$\left(J_{AB}^{\mathcal{N}}\right)^{T_{B}}\geq0$. PPT states and
operations can be regarded as free; therefore, anything that is \emph{not}
PPT---which we call NPT---will be a resource. For this reason, this
resource theory is called the resource theory of NPT entanglement.
As it often happens in resource theories, one considers a larger set
of operations to get upper and lower bounds for the relevant figures
of merit, especially when the interesting resource theory is mathematically
hard to study, such as the LOCC theory \citep{NP-hard1,NP-hard2}.
This is precisely why we study NPT entanglement.

Here we generalize partial transpose, defining the transpose supermap
$\Upsilon$ as $\Upsilon\left[\mathcal{N}_{A}\right]=\mathcal{T}_{A_{1}}\circ\mathcal{N}_{A}\circ\mathcal{T}_{A_{0}}$.
Note that the Choi matrix of $\Upsilon\left[\mathcal{N}_{A}\right]$
is the transpose of the Choi matrix of $\mathcal{N}_{A}$. In this
way, PPT channels can be characterized in a similar way to PPT states:
PPT channels are those bipartite channels such that $\Upsilon_{B}\left[\mathcal{N}_{AB}\right]$
is still a valid channel. Now we iterate the previous construction
to define PPT superchannels.
\begin{defn}
\label{def:PPT}A superchannel $\Theta_{AB\rightarrow A'B'}$ is \emph{PPT}
if $\Theta_{AB\rightarrow A'B'}^{\Gamma}:=\Upsilon_{B'}\circ\Theta_{AB\rightarrow A'B'}\circ\Upsilon_{B}$
is still a valid superchannel.
\end{defn}

These superchannels enjoy some remarkable properties.
\begin{lem}
The following are equivalent:
\begin{enumerate}
\item $\Theta_{AB\rightarrow A'B'}$ is a PPT superchannel.
\item $\Theta_{AB\rightarrow A'B'}$ is completely PPT preserving.\label{enu:complete}
\item $\left(\mathbf{J}_{ABA'B'}^{\Theta}\right)^{T_{BB'}}\geq0$, where
$\mathbf{J}_{ABA'B'}^{\Theta}$ is the Choi matrix of the superchannel
$\Theta$ \citep{Circuit-architecture,Hierarchy-combs,Perinotti1,Perinotti2,Gour2018}.\label{enu:PPT Choi}
\end{enumerate}
\end{lem}

A proof of this result can be found in Ref.~\citep[subsection 5 A]{Gour-Scandolo}.
Property~\ref{enu:complete} means that PPT superchannels preserve
PPT channels in a complete sense. Property~\ref{enu:PPT Choi} tells
us that PPT superchannels are the same objects that appeared in Ref.~\citep{Leung}.
Despite the fairly simple condition defining PPT superchannels at
the level of Choi matrices, we do not know if all of them can be realized
with PPT pre- and postprocessing. When this happens, we call them
\emph{restricted} PPT superchannels. It is not hard to show that restricted
PPT superchannels are indeed PPT superchannels in the sense of definition~\ref{def:PPT}.
Instead, we conjecture that the converse is not true, so we are really
considering a larger set of superchannels. This is one of the main
differences from a related work by Wang and Wilde \citep{WW18}: there
the authors study only restricted PPT superchannels, and they do not
consider bipartite channels, but only one-way channels from Alice
to Bob (or vice versa).

Our approach brings a lot of mathematical simplifications. For instance,
if we replace LOCC with PPT in Eqs.~\eqref{eq:complete family} and
\eqref{eq:conversion distance}, the NPT entanglement measures and
the conversion distance become computable efficiently with SDPs (see
Ref.~\citep[subsections 5 B and 5 C]{Gour-Scandolo}). However, note
that this family of NPT entanglement monotones will not provide a
sufficient condition for the convertibility under \emph{LOCC} superchannels.

\paragraph{New measures of dynamical entanglement.}

Since PPT channels contain LOCC channels, PPT superchannels contain
LOCC ones. Thus, measures of NPT dynamical entanglement (i.e., monotonic
under PPT superchannels) are also measures of LOCC dynamical entanglement
(i.e., monotonic under LOCC superchannels). As seen above, working
with PPT superchannels is mathematically simpler. For this reason,
focusing on the PPT-simulation of channels we obtain measures of LOCC
dynamical entanglement that are easily computable.

The first example in this respect is the the \emph{negativity} \citep{VW2002},
defined for states as $N\left(\rho_{AB}\right)=\frac{\left\Vert \mathcal{T}_{B}\left(\rho_{AB}\right)\right\Vert _{1}-1}{2}$.
The generalization to bipartite channels is straightforward: replace
the trace norm with the diamond norm, and the transpose map $\mathcal{T}_{B}$
with the transpose supermap $\Upsilon_{B}$.
\begin{equation}
N\left(\mathcal{N}_{AB}\right)=\frac{\left\Vert \Upsilon_{B}\left[\mathcal{N}_{AB}\right]\right\Vert _{\diamond}-1}{2}.
\end{equation}
Contextually, the logarithmic negativity is defined as 
\begin{equation}
LN\left(\mathcal{N}_{AB}\right)=\log_{2}\left\Vert \Upsilon_{B}\left[\mathcal{N}_{AB}\right]\right\Vert _{\diamond}.
\end{equation}
We prove that these are measures of dynamical entanglement that can
be computed efficiently with an SDP (cf.\ Ref.~\citep[subsection 5 C]{Gour-Scandolo}).

Now we introduce a new measure of NPT dynamical entanglement, called
max-logarithmic negativity (MLN) (cf.\ \citep{WW19}). It is a generalization
of the notion of $\kappa$-entanglement introduced in \citep{WW18}.
The MLN is defined as
\begin{align}
 & LN_{\max}\left(\mathcal{N}_{AB}\right)\nonumber \\
 & :=\log_{2}\inf_{P_{AB}}\left\{ \max\left\{ \left\Vert P_{A_{0}B_{0}}\right\Vert _{\infty},\left\Vert P_{A_{0}B_{0}}^{T_{B_{0}}}\right\Vert _{\infty}\right\} \right\} ,\label{eq:MLN}
\end{align}
where $P_{AB}$ is a matrix subject to the constraints $-P_{AB}^{T_{B}}\leq\left(J_{AB}^{\mathcal{N}}\right)^{T_{B}}\leq P_{AB}^{T_{B}}$
and $P_{AB}\geq0$. Here $P_{A_{0}B_{0}}$ denotes $\mathrm{Tr}_{A_{1}B_{1}}\left[P_{AB}\right]$.
We can show that the MLN is an additive measure of dynamical entanglement,
computable with an SDP (see Ref.~\citep[subsection 5 C]{Gour-Scandolo}).

Despite its rather complicated definition, the MLN has a nice operational
interpretation, which generalizes the results in Refs.~\citep{PPTcost1,WW18}.
Consider the task of simulating $n$ parallel copies of the bipartite
channel $\mathcal{N}_{AB}$ out of the maximally entangled state $\left|\phi_{m}^{+}\right\rangle _{A_{0}B_{0}}$
of Schmidt rank $m$ using PPT superchannels (which, in this case,
take the form of PPT channels). Recall that $\left|\phi_{m}^{+}\right\rangle _{A_{0}B_{0}}$
is, up to a scaling factor 2, the maximal resource in the theory of
entanglement for bipartite channels, as we noted above. We require
that the conversion of $\left|\phi_{m}^{+}\right\rangle _{A_{0}B_{0}}$
into $\mathcal{N}_{AB}$ be exact for every $n$.  We want to study
the asymptotic entanglement cost of preparing $\mathcal{N}_{AB}$
according to this PPT protocol, viz.\ the minimum Schmidt rank of
maximally entangled states consumed per copy of $\mathcal{N}_{AB}$
produced when $n\rightarrow+\infty$. Remarkably, we show that this
cost is given precisely by the MLN. Clearly, the use of PPT superchannels
is not so physically motivated, but it provides a simple lower bound
to the more meaningful calculation of the entanglement cost under
LOCC superchannels \citep{PPTcost1,WW18}.
\begin{thm}
\label{thm:The-exact-asymptotic}The exact asymptotic NPT cost of
a bipartite channel $\mathcal{N}_{AB}$ is \textup{$LN_{\max}\left(\mathcal{N}_{AB}\right)$.}
\end{thm}

A proof of this result can be found in Ref.~\citep[subsection 5 D]{Gour-Scandolo}.
We can prove that the MLN is an upper bound for another entanglement
measure, the NPT entanglement generation power $E_{g}^{\mathrm{PPT}}$
\citep{Bennett-bipartite,Resource-channels-1,Resource-channels-2,Gour-Winter}
(cf.\ Appendix~\ref{sec:Bound-between-NPT}):
\[
E_{g}^{\mathrm{PPT}}\left(\mathcal{N}_{AB}\right)\leq LN_{\max}\left(\mathcal{N}_{AB}\right).
\]

\paragraph{Bound entanglement for bipartite channels.}

Dual to the calculation of the cost of a bipartite channel, we have
the distillation of ebits out of a dynamical resource. It is known
that for some entangled static resources this is not possible: it
is the phenomenon of bound entanglement \citep{Bound-entanglement},
which occurs whenever we have a PPT entangled state.

Is it possible to distill ebits out $n$ copies of a PPT channel $\mathcal{N}_{AB}$?
Now, when we have $n$ copies of a channel, the \emph{timing} in which
they are available becomes relevant: dynamical resources have a natural
temporal ordering between input and output. Indeed, unlike states,
they can also be composed in nonparallel ways, e.g., in sequence.
Therefore, when manipulating dynamical resources, we also need to
specify \emph{when} and \emph{how} they can be used (see also Refs.~\citep{Resource-channels-2,Gour-Scandolo}).
This opens up the possibility of using adaptive schemes \citep{Pirandola-LOCC,Kaur2017,Wilde-cost,Wilde-entanglement}:
if we have $n$ resources $\mathcal{N}_{1},\dots,\mathcal{N}_{n}$
that are available, respectively, at times $t_{1}\leq t_{2}\leq\dots\leq t_{n}$,
the most general channel that can be simulated with these resources
is given by a free $n$-comb \citep{Gutoski,Circuit-architecture,Hierarchy-combs,Gutoski2,Resource-theories,Gutoski3,Resource-channels-1,Resource-channels-2},
depicted in Fig.~\ref{fig:A-PPT--comb} in the case of a PPT comb.
\begin{figure}
\begin{centering}
\includegraphics[width=1\columnwidth]{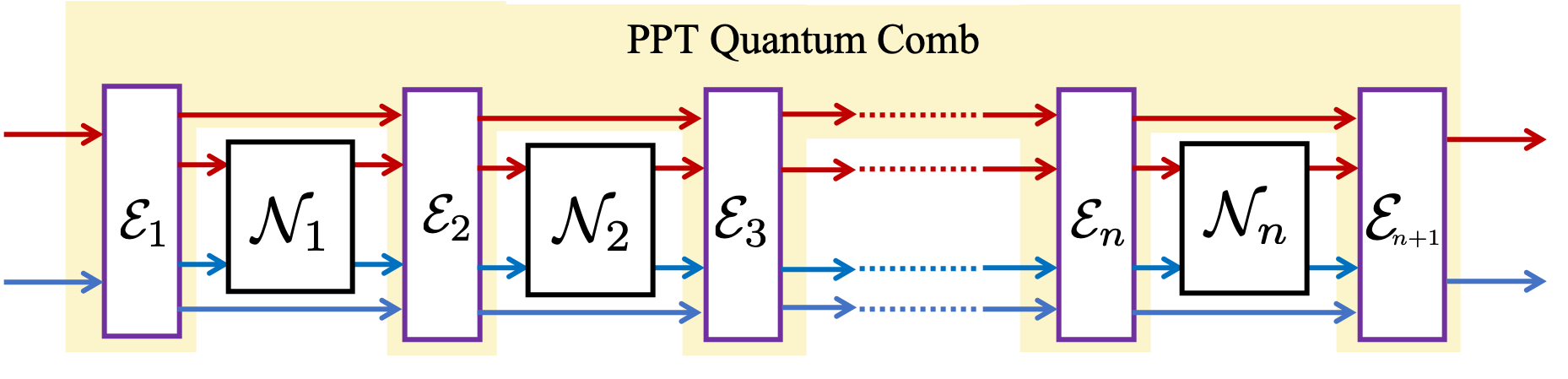}
\par\end{centering}
\caption{\label{fig:A-PPT--comb}A PPT $n$-comb acts on $n$ bipartite channels
$\mathcal{N}_{j}$, e.g., to distill ebits.}
\end{figure}
 Specializing this idea to the case of dynamical entanglement, this
amounts to considering an LOCC $n$-comb, where all the $n+1$ channels
$\mathcal{E}_{1}$, \ldots , $\mathcal{E}_{n+1}$ in Fig.~\ref{fig:A-PPT--comb}
are LOCC. Then we plug the $n$ copies of $\mathcal{N}_{AB}$ into
its $n$ slots.

Instead of LOCC combs, we consider PPT combs, which are defined as
the combs for which the composition of channels $\mathcal{E}_{n+1}\circ\dots\circ\mathcal{E}_{1}$
in Fig.~\ref{fig:A-PPT--comb} is a PPT channel. This is equivalent
to requiring that the Choi matrix of the $n$-comb \citep{Circuit-architecture,Hierarchy-combs,Gour-Scandolo}
is the Choi matrix of a PPT channel. PPT combs will give us an upper
bound on the amount of ebits generated in an LOCC procedure. However,
again, we do not know if this implies that each channel $\mathcal{E}_{1}$,
\ldots , $\mathcal{E}_{n+1}$ is PPT, but we conjecture it is not
the case. 

By the mathematical properties of PPT combs and PPT channels, we can
show that no ebits can be distilled out of PPT channels even with
the most general adaptive PPT scheme (see Ref.~\citep[section 7]{Gour-Scandolo}).
Since this is an upper bound for LOCC adaptive schemes, we conclude
that no entanglement distillation from PPT channels is possible under
LOCC protocols either.
\begin{thm}
It is impossible to distill entangled ebits from PPT channels under
any adaptive schemes in any resource theory of dynamical entanglement.
\end{thm}

As a result, we find an example of a bound entangled POVM.
\begin{example}
Recall that a POVM can be viewed as a quantum-to-classical channel.
Let $\beta_{A_{0}B_{0}}$ be any PPT bound entangled state of a bipartite
system $A_{0}B_{0}$, and consider the binary POVM $\left\{ \beta_{A_{0}B_{0}},I_{A_{0}B_{0}}-\beta_{A_{0}B_{0}}\right\} $.
Since both $\beta_{A_{0}B_{0}}$ and $I_{A_{0}B_{0}}-\beta_{A_{0}B_{0}}$
have positive partial transpose, it follows that this POVM is a PPT
channel. As such, it cannot produce distillable entanglement. This
means it is a bound entangled POVM.
\end{example}

\paragraph{Conclusions and outlook.}

In this letter, we addressed dynamical entanglement as a resource
theory of quantum processes. This is a major step in understanding
the role of entanglement in quantum theory, for it allows us to treat
static and dynamical entanglement on the same grounds \citep{Devetak-Winter,Resource-calculus},
which is something that had been missing since the inception of the
very first quantum information protocols \citep{Teleportation,Dense-coding}.
We found a set of measures of dynamical entanglement yielding necessary
and sufficient conditions for LOCC channel simulation. Then we generalized
the key tool of partial transpose, defining PPT superchannels. Working
with them, we obtained measures of dynamical entanglement that can
be computed with SDPs. This remarkable fact, which did not appear
in previous works on PPT superchannels (e.g., Ref.~\citep{WW18}),
is a consequence of our more relaxed definition of PPT superchannels
(definition~\ref{def:PPT}). This is not the only novelty with respect
to Ref.~\citep{WW18}: we were able to generalize their notion of
$\kappa$-entanglement with the max-logarithmic negativity (Eq.~\eqref{eq:MLN}).
Finally, we showed that we can distill no ebits under any adaptive
strategies out of PPT channels. This extends the known result for
PPT states \citep{Bound-entanglement}, and led us to the discovery
of bound entangled POVMs.

Clearly, our work just scratches the surface of a whole unexplored
world, opening the way for a thorough study of the new area of dynamical
entanglement. On a grand scale, our findings lead naturally to several
directions that can be explored anew. Think, e.g., of multipartite
entanglement \citep{Review-entanglement}, or of the whole zoo of
entanglement measures \citep{Plenio-review,Review-entanglement},
to be extended to channels. Moreover, our results for LOCC superchannels
can be translated to local operations and shared randomness (LOSR)
superchannels \citep{Rosset-distributed,Wolfe2020,Schmid2020,LOSR-nonlocality},
which are a strict subset of LOCC ones. LOSR superchannels were proved
essential for the formulation of resource theories for non-locality
\citep{LOSR-nonlocality}: they define the relevant notion of dynamical
entanglement in Bell and common-cause scenarios. This intriguing research
direction deserves a comprehensive study in the future.

Finally, providing us with a more general angle, research findings
in the resource theory of dynamical entanglement can also help us
gain new insights into one of the major open problems of quantum information
theory: the existence of NPT bound entangled states \citep{NPPT1,NPPT2,NPPT3}.\medskip{}

\begin{acknowledgments}
G.\ G.\ would like to thank Francesco Buscemi, Eric Chitambar, Mark
Wilde, and Andreas Winter for many useful discussions related to the
topic of this paper. The authors acknowledge support from the Natural
Sciences and Engineering Research Council of Canada (NSERC) through
grant RGPIN-2020-03938, from the Pacific Institute for the Mathematical
Sciences (PIMS), and a from Faculty of Science Grand Challenge award
at the University of Calgary.
\end{acknowledgments}

\bibliographystyle{apsrev4-2}
\bibliography{PPT}

\onecolumngrid

\appendix

\section{\label{sec:Bound-between-NPT}Bound between NPT entanglement generation
power and max-logarithmic negativity}

Define the \emph{NPT entanglement generation power} \citep{Bennett-bipartite,Resource-channels-1,Resource-channels-2,Gour-Winter}
as the maximum amount of NPT entanglement produced out of PPT states,
namely as
\begin{equation}
E_{g}^{\mathrm{PPT}}\left(\mathcal{N}_{AB}\right)=\max_{\rho_{A_{0}'B_{0}'A_{0}B_{0}}\textrm{ PPT}}E\left(\mathcal{N}_{AB}\left(\rho_{A_{0}'B_{0}'A_{0}B_{0}}\right)\right),\label{eq:e-generation power}
\end{equation}
where $E$ is a measure of NPT static entanglement.

In~\citep{Gour-Scandolo}, we defined the \emph{exact} NPT entanglement
single-shot distillation out of a bipartite channel $\mathcal{N}_{AB}$
as
\[
\mathrm{DISTILL}_{\mathrm{PPT},\mathrm{exact}}^{\left(1\right)}\left(\mathcal{N}_{AB}\right)=\log_{2}\max\left\{ m:\mathcal{N}_{AB}\overset{\mathrm{PPT}}{\longrightarrow}\phi_{m}^{+}\right\} ,
\]
where $\phi_{m}^{+}$ is the maximally entangled state with Schmidt
rank $m$. The distillable NPT entanglement in the asymptotic limit
is then
\[
\mathrm{DISTILL}_{\mathrm{PPT},\mathrm{exact}}\left(\mathcal{N}_{AB}\right)=\limsup_{n}\frac{1}{n}\mathrm{DISTILL}_{\mathrm{PPT},\mathrm{exact}}^{\left(1\right)}\left(\mathcal{N}_{AB}^{\otimes n}\right),
\]
where we are using a parallel scheme for distillation. If $E$ in
Eq.~\eqref{eq:e-generation power} is taken to be the exact asymptotic
PPT distillation of static entanglement \citep{Resource-channels-1},
we can link the NPT entanglement generation power to the MLN.
\begin{lem}
For any bipartite channel $\mathcal{N}_{AB}$, we have
\[
E_{g}^{\mathrm{PPT}}\left(\mathcal{N}_{AB}\right)\leq\mathrm{DISTILL}_{\mathrm{PPT},\mathrm{exact}}\left(\mathcal{N}_{AB}\right),
\]
where $E_{g}^{\mathrm{PPT}}$ is defined using the exact asymptotic
PPT distillation of static entanglement.
\end{lem}

\begin{proof}
The proof follows similar lines to the proof of theorem~4 in~\citep{Resource-channels-1}.
Let $R=E_{g}^{\mathrm{PPT}}\left(\mathcal{N}_{AB}\right)$, and let
$\omega_{A_{0}'B_{0}'A_{0}B_{0}}$ be the optimal PPT state achieving
$R$, i.e., $E\left(\mathcal{N}_{AB}\left(\omega_{A_{0}'B_{0}'A_{0}B_{0}}\right)\right)=R$.
Now let us construct a distillation protocol for $\phi_{m}^{+}$.
To this end, consider the channel preparing $\omega_{A_{0}'B_{0}'A_{0}B_{0}}^{\otimes n}$
out of nothing (namely out of the 1-dimensional system). This is a
PPT channel, which we will use as a pre-processing to construct a
(restricted) PPT superchannel. Defining $\sigma_{A_{0}'B_{0}'A_{1}B_{1}}:=\mathcal{N}_{AB}\left(\omega_{A_{0}'B_{0}'A_{0}B_{0}}\right)$,
then $\mathcal{N}_{AB}^{\otimes n}\left(\omega_{A_{0}'B_{0}'A_{0}B_{0}}^{\otimes n}\right)=\sigma_{A_{0}'B_{0}'A_{1}B_{1}}^{\otimes n}.$
Since $E_{g}^{\mathrm{PPT}}$ is defined as the exact asymptotic distillation
rate, we know that there exists a PPT post-processing such that $\limsup_{n}\frac{1}{n}\mathrm{DISTILL}_{\mathrm{PPT},\mathrm{exact}}^{\left(1\right)}\left(\sigma_{A_{0}'B_{0}'A_{1}B_{1}}^{\otimes n}\right)=R$,
where $\mathrm{DISTILL}_{\mathrm{PPT},\mathrm{exact}}^{\left(1\right)}$
is the analogous definition for states rather than channels. With
this in mind, we can use this post-processing to obtain some $\phi_{m}^{+}$;
thus we prove our statement.
\end{proof}
By a Carnot-like argument \citep{Thermodynamic-entanglement}, one
can prove that $\mathrm{DISTILL}_{\mathrm{PPT},\mathrm{exact}}\left(\mathcal{N}_{AB}\right)\leq\mathrm{COST}_{\mathrm{PPT},\mathrm{exact}}\left(\mathcal{N}_{AB}\right)$,
where $\mathrm{COST}_{\mathrm{PPT},\mathrm{exact}}\left(\mathcal{N}_{AB}\right)$
denotes the exact asymptotic NPT cost of $\mathcal{N}_{AB}$. By theorem~4
in the main letter, the exact asymptotic NPT cost of $\mathcal{N}_{AB}$
is the MLN, whence we conclude that
\[
E_{g}^{\mathrm{PPT}}\left(\mathcal{N}_{AB}\right)\leq LN_{\max}\left(\mathcal{N}_{AB}\right).
\]

\end{document}